\documentclass[a4paper,USenglish]{lipics-v2018}
\usepackage{booktabs}
\usepackage{thmtools}
\usepackage{thm-restate}
\usepackage{hyperref}
\usepackage{cleveref}
\nolinenumbers

\title{
Designing Practical PTASes for Minimum Feedback Vertex Set in Planar Graphs
}

\titlerunning{Designing Practical PTASes for Minimum Feedback Vertex Set in Planar Graphs}

\author{Glencora Borradaile}{Oregon State University \\{[Corvallis, OR, USA]}}{glencora@eecs.oregonstate.edu}{}{}

\author{Hung Le}{Oregon State University \\{[Corvallis, OR, USA]}}{lehu@oregonstate.edu}{}{}

\author{Baigong Zheng}{Oregon State University \\{[Corvallis, OR, USA]}}{zhengb@oregonstate.edu}{}{}

\authorrunning{G.\, Borradaile, H.\, Le and B.\, Zheng}

\Copyright{Glencora Borradaile, Hung Le and Baigong Zheng}

\subjclass{ Theory of computation $\rightarrow$ Graph algorithms analysis}


\category{}

\relatedversion{}

\supplement{Code available at:  https://github.com/zbgzbg2007/Practical-FVS-PTASes}

\funding{
This material is based upon work supported by the National Science Foundation under Grant No.\ CCF-1252833.
}

\acknowledgements{ The authors would like to thank Shay Mozes and Eli Fox-Epstein for sharing their code~\cite{FMPS16} to preprocess the road network graphs.}

\EventEditors{John Q. Open and Joan R. Access}
\EventNoEds{2}
\EventLongTitle{42nd Conference on Very Important Topics (CVIT 2016)}
\EventShortTitle{CVIT 2016}
\EventAcronym{CVIT}
\EventYear{2016}
\EventDate{December 24--27, 2016}
\EventLocation{Little Whinging, United Kingdom}
\EventLogo{}
\SeriesVolume{42}
\ArticleNo{23}

\begin{document}

\maketitle

\begin{abstract}
We present two algorithms for the minimum feedback vertex set problem in planar graphs: an $O(n \log n)$ PTAS using a linear kernel and balanced separator, and a heuristic algorithm using kernelization and local search.  We implemented these algorithms and compared their performance with Becker and Geiger's 2-approximation algorithm~\cite{BG96}. We observe that while our PTAS is competitive with the 2-approximation algorithm on large planar graphs, its running time is much longer.
And our heuristic algorithm can outperform the 2-approximation algorithm on most large planar graphs and provide a trade-off between running time and solution quality, i.e. a ``PTAS behavior''.
 \end{abstract}

\section{Introduction}

The {\em minimum feedback vertex set} problem (FVS) asks for a minimum set of vertices in an undirected graph such that after removing this set the resulting graph has no cycle.
This problem is one of Karp's 21 original NP-Complete problems~\cite{Karp72} and has applications in different areas, such as deadlock recovery in operating systems and reducing computation in Bayesian inference~\cite{BGNR98}. 
The current best approximation ratio for FVS in general graphs is 2 due to Becker and Geiger~\cite{BG96} and 
Bafna, Berman and Fujito~\cite{BBF99}, both of which could also work for the vertex-weighted version of FVS.
Chudak et al.~\cite{CGHW98} showed that the two algorithms can be explained in terms of the primal dual method and simplified the latter algorithm.

In planar graphs, the problem is still NP-hard~\cite{Yannakakis78}, but we can obtain better approximation algorithms.  A polynomial-time approximation scheme (PTAS) is a $(1+\epsilon)$-approximation algorithm runs in polynomial time for any fixed $\epsilon > 0$.  As far as we know, Kleinberg and Kumar~\cite{KK01} gave the first PTAS for FVS in planar graphs.  Demaine and Hajiaghayi~\cite{DH05} gave a different PTAS for FVS in single-crossing-minor-free graphs through their bidimensionality theory.  Cohen-Addad et al.~\cite{CCKMM16} gave a PTAS for the vertex-weighted version of this problem in bounded-genus graphs.  In our companion work~\cite{LZ18}, we showed that local search is also a PTAS for FVS in minor-free graphs.

Besides the above algorithms, researchers also proposed some heuristic algorithms for FVS and evaluated their performance in experiments.
For example, Pardalos et al.~\cite{PQR98} developed a greedy randomized adaptive search procedure for FVS, Brunetta et al.~\cite{BMT00} proposed a system based on local search and a branch-and-cut algorithm, Zhang et al.~\cite{ZYZS13} presented a variable depth-based local search algorithm with a randomized scheme, and Qin and Zhou~\cite{QZ14} introduced a simulated annealing local search algorithm for FVS.
However, all of these works focus on general graphs. Brunetta et al.~\cite{BMT00} included planar graphs in their experiments, but their test planar graphs were not very large, having at most one thousand of vertices.
So it is natural to ask the following question: 
\begin{center}
{\em which algorithm is preferred for FVS on large planar graphs in practice?}
\end{center}
One potential answer to this question is a PTAS.   
There are two reasons supporting this choice: theoretically, PTASes can provide the best approximation ratio; practically, it has been shown that PTASes can be made practical for the minimum dominating set problem~\cite{MGJ09}, the Steiner tree problem~\cite{TM11} and the traveling salesman problem (TSP)~\cite{BFKM17} in large planar graphs.
Unfortunately, we find that a simple PTAS for FVS does not find a more accurate solution than the 2-approximation algorithm in most real-world graphs and some synthetic graphs. For those test graphs where it can find better solutions, the improvement is less than 1 percent. Furthermore, the PTAS is much slower than the 2-approximation algorithm.
So another question to ask is 
\begin{center}
{\em can we obtain a practical algorithm that can find more accurate solutions than \\ the 2-approximation algorithm in large planar graphs?}
\end{center}
In this work, we propose a heuristic algorithm and show that it can outperform the 2-approximation algorithm in real-world graphs and most large synthetic graphs.

\subsection{Overview of Our Work}
To answer the first question, we implemented Becker and Geiger's 2-approximation algorithm~\cite{BG96} as our baseline, which is simpler than the algorithm of Bafna, Berman and Fujito~\cite{BBF99}.
In Section~\ref{sec:2app-opt}, we evaluate this implementation on some graphs where we can obtain optimal solutions or good lower bound for the optimal solutions.
We find that the 2-approximation algorithm finds solutions that are very close to the optimal in these instances.

To outperform this baseline, in Section~\ref{sec:ptas} we propose a simple-to-implement $O(n \log n)$ PTAS for FVS in planar graphs, which starts with a linear kernel for FVS (see Section~\ref{sec:kernel}) and then uses a balanced separator (see Section~\ref{sec:sep}) applied recursively to decompose the graph into a set of small subgraphs in which we will solve the problem optimally. 
The approach based on balanced separators has been applied to obtain PTASes for the maximum independent set problem~\cite{LT80} and the minimum vertex cover problem~\cite{CNS81} in planar graphs. However, this approach is criticized in literature for two reasons: (1) a good approximation ratio can only be obtained in very large graphs~\cite{CNS82, Baker94, DH05}, and (2) it needs the size of the optimal solution to be linear w.r.t. the size of input graph~\cite{Grohe03}. 
We overcome both issues.
For the first, we relate the error parameter $\epsilon$ to the largest size of the decomposed graphs instead of the size of the original graph, and the algorithm can provide good approximation ratio for any graph in this way.
For the second, we use a linear kernel as a proxy to achieve the linear bound for the optimal solution.
Since many problems~\cite{FLST10, FLST12, BFLPST16} admit linear kernels in minor-free graphs, which is a more general graph family and admits balanced separators of sublinear size~\cite{AST90}, we believe this approach could be applied more generally.

Other PTASes for planar FVS that have been proposed are either complicated to implement (relying on dynamic programming over tree decompositions) or not sufficiently efficient (having running time of the form $n^{O(1/\epsilon)}$).
Compared to those, our PTAS has some obvious advantages: (1) it only relies on some simple algorithmic components like the kernelization algorithm, which consists of a sequence of simple reduction rules, and balanced separators, which are known to be practical~\cite{ADGM07, HSWPZ09, FMPS16}; (2) it has very few parameters to optimize; (3) the constants behind the big $O$ notation are small enough and (4) its running time is theoretically faster than previous PTASes.
Performance of this PTAS on different large planar graphs is discussed in Section~\ref{sec:ptas-2app}.
We also incorporated heuristic steps (Section~\ref{sec:heu-ptas}) to improve its solution and analyzed the influence of the parameters of the heuristic.
However, counter to the success stories for the Steiner tree and TSP PTASes, we find that the solution found by this PTAS does not outperform the precision of the 2-approximation algorithm significantly.

Although our PTAS does not outperform the 2-approximation algorithm significantly, 
we use it as inspiration to engineer a PTAS-like heuristic we call a {\em Heuristic Approximation Scheme (HAS)}, that is a heuristic with a running-time/precision trade-off.
Our HAS has two main steps.
The first step (see Section~\ref{sec:hybrid}) is a hybrid algorithm that alternates the reduction rules of the linear kernel and the greedy step of the 2-approximation algorithm. 
The second step (see Section~\ref{sec:ls}) is a variant of local search.
Many local search heuristics start with a feasible solution and repeatedly construct a smaller solution by replacing a subset $A$ of the original solution with another smaller subset of the non-solution vertices, if it is feasible.
We notice that this could be inefficient, since there are too many ways to replace the subset $A$ and so only very small values of $|A|$ can be handled.
Instead, we use a {\em fixed-parameter tractable} (FPT) algorithm to determine the replacement for $A$.
An algorithm is FPT if it can solve a given problem optimally in running time $f(k) \cdot n^{O(1)}$, where $k$ is given as a fixed parameter (such as the size of the optimal solution, as for FVS) and $f$ is an arbitrary computable function. 
This kind of algorithm is very efficient when the parameter $k$ is small. 
Now given a feasible solution, our local search heuristic will repeatedly improve the solution by selecting a set $A$ from the solution, constructing a graph as the union of the non-solution vertices and set $A$, solving the problem in this graph optimally with the FPT algorithm, and replacing $A$ with the obtained optimal solution.

We implemented and evaluated HAS on different large planar graphs and analyzed the effects of its parameters (Section~\ref{sec:heu-2app}).
Our result shows that even its first step is able to find better solutions than the 2-approximation algorithm on most of our test graphs and its second step improves these solutions further.
As a result, the total improvement for all real-world graphs is at least 5 percent, which is more than 30000 vertices in the largest test graph.

We find that HAS is very flexible and provides a kind of ``PTAS behavior''.  
Its first step is competitive w.r.t. the running time so we can obtain a good solution quickly.
And its second step can be applied for a longer time to obtain a better solution when the running time is not strictly limited.
Thus, it can provide a trade-off between the running time and the solution quality by choices of the number of local search iterations.
Therefore, we believe this algorithm will be a better choice in practice.

\section{The Algorithms for FVS in Planar Graphs}
In this section, we briefly summarize the FVS algorithms we implemented for planar graphs:
the 2-approximation algorithm of Becker and Geiger~\cite{BG96}, 
Bonamy and Kowalik's linear kernel~\cite{BK16} for FVS in planar graphs (with optimizing modifications that we designed), 
our new PTAS using this linear kernel and balanced planar separators, 
and our proposed Heuristic Approximation Scheme.

\subsection{The 2-Approximation Algorithm}
Becker and Geiger's 2-approximation algorithm~\cite{BG96} works for vertex-weighted FVS in general graphs and
consists of two steps: (1) computes a greedy solution and (2) removes vertices to obtain a minimal solution.
In the first step, the algorithm assigns a score for each vertex (weight of the vertex divided by its degree), and repeatedly removes a vertex with minimum score from the graph and adds it to the greedy solution.
Each time a vertex is removed, the scores of its neighbors are updated.
In the second step, the algorithm tries to remove the vertices from the greedy solution in the reverse order in which they were added, to obtain a minimal feasible solution.

\subsection{Kernelization Algorithm}\label{sec:kernel}

A parameterized decision problem with a parameter $k$ admits a {\em kernel} if there is a polynomial time algorithm (where the degree of the polynomial is independent of $k$), called a {\em kernelization algorithm}, that outputs a decision-equivalent instance whose size is bounded by some function $h(k)$.
If the function $h(k)$ is linear in $k$, then we say the problem admits a {\em linear} kernel.

Bonamy and Kowalik's linear kernel for planar FVS~\cite{BK16} consists of a sequence of 17 reduction rules.
Each rule replaces a particular subgraph with another (possibly empty) subgraph, and possibly marks some vertices that must be in an optimal solution. 
The first 12 rules are simple and sufficient to obtain a $15k$-kernel~\cite{BK16}.
Since the remaining rules do not improve the kernel by much, and since Rule 12 is a rejecting rule\footnote{This is to return a trivial no-instance for the decision problems when the resulting graph has more than $15k$ vertices.}, we only implement the first 11 rules, all of which are local and independent of the parameter $k$.
The algorithm starts by repeatedly applying the first five rules to the graph and initializes two queues: queue $Q_1$ contains some vertex pairs that are candidates to check for Rule 6 and queue $Q_2$ contains vertices that are candidates to check for the last five rules.
While $Q_1$ is not empty, the algorithm repeatedly applies Rule 6, reducing $|Q_1|$ in each step.
Then the algorithm repeatedly applies the remaining rules in order, reducing $|Q_2|$ until $Q_2$ is empty.
After applying any rule, the algorithm updates both queues as necessary, and will
apply the first five rules if applicable. 
See the original paper~\cite{BK16} for full details of the reduction rules.

We remark that in the original paper~\cite{BK16}, the algorithm runs in expected $O(n)$ time, and each rule can be detected in $O(1)$ time for each candidate with a hash table.  However, we choose to use a balanced binary search tree instead of a hash table for a better practical performance.

The original algorithm works for the decision problem, so when applying it to the optimization problem, we need an additional step, called {\em lifting}, to convert a kernel solution to a feasible solution for the original graph. 
If a reduction rule does not introduce new vertices, then the lifting step for it will be trivial.
Otherwise, we need to handle the vertices introduced by reduction steps, if they appear in the kernel solution.
Among all the reduction rules, there are two rules, namely Rule 8 and Rule 9, that will introduce new vertices into the graph.  
For Rule 8, the following observation, whose proof is provided in the appendix, shows that the new vertex can be avoided.

\begin{restatable}{observation}{obsy}
\label{obs:y}
The new vertex introduced by Rule 8 can be replaced by a vertex from the original graph.
\end{restatable}

For Rule 9, we record in a structure the related vertices for each application of this rule, and store the structures in a list in the same order as we apply the rule.
To lift the solution to the original graph, we store all the vertices of the solution in a balanced search tree and then check the structures in the reverse order to see if the recorded vertices in a structure are also in the solution.  
If there are involved vertices in the solution, we modify the solution according to the reversed Rule 9. Since Rule 9 decreases the size of the graph, it can be applied at most $O(n)$ times.
So there are at most $O(n)$ structures to check, each of which contains only constant number of vertices.
To check if a vertex in the solution we need $O(\log n)$ time, so the total replacement can be done in $O(n \log n)$ time.
Then the lifting step first add back all vertices marked by the kernelization algorithm which can be done in linear time and then replace the introduced new vertices with original vertices, which needs at most $O(n \log n)$ time. So lifting can be accomplished in $O(n \log n)$ time.

\subsection{Polynomial-Time Approximation Scheme}\label{sec:ptas}
In this section, we introduce a PTAS for FVS in planar graphs, using linear kernel and balanced separator.

\subsubsection{Balanced Separator}\label{sec:sep}

A {\em separator} is a set of vertices, removing which will partition the graph into two parts.  A separator is {\em $\alpha$-balanced} if those two parts both have at most an $\alpha$-fraction of the original vertex set.  Lipton and Tarjan~\cite{LT79} first introduced a separator theorem for planar graphs, which says a planar graph with $n$ vertices admits a $\frac{2}{3}$-balanced separator of size at most $2\sqrt{2n}$, and they gave a linear-time algorithm to compute such a balanced separator.
This algorithm starts by computing a breadth-first search (BFS) tree for the graph, partitioning the vertices into levels. Then it tries to find the separator in three phases: 
\begin{enumerate}[(P1)] 
\item if there is any BFS level satisfying the requirements, then it is returned as a result; 
\item if there is no such level, then the algorithm tries to find two BFS levels that can form a balanced separator together;
\item if both previous phases fail, then the algorithm identifies two BFS levels and constructs a fundamental cycle\footnote{Given a spanning tree for a graph, a fundamental cycle consists of a non-tree edge and a path in the tree connecting the two endpoints of that edge.} to separate the subgraph between these two levels, such that the union of these two levels and the fundamental cycle form a balanced separator.
\end{enumerate}

Though we followed a textbook version~\cite{Kozen12} of this separator algorithm, our implementation still guarantees the $2\sqrt{2n}$ bound for the size of the separator.  
We remark that we did not apply heuristics in our implementation for the separator algorithm.  This is because we did not observe separator size improvement by some simple heuristics in the early stage of this work, and these heuristics may slow down the separator algorithm.  Since our test graphs are large (up to 2 million vertices) and we will apply the algorithm recursively in our PTAS, these heuristics may slow down our PTAS even more.

\subsubsection{PTAS for Planar FVS}
Lipton and Tarjan~\cite{LT80} designed a PTAS for maximum independent set in planar graphs using their balanced separator, which depends on the fact that the input is already a constant-factor approximation to the maximum independent set and contains the optimal solution.
Here we use the linear kernel as a proxy for the constant factor approximation that can be used to obtain a nearly optimal solution for FVS and relate the error parameter $\epsilon$ to the largest size of decomposed graphs instead of the size of the input graph as previous works~\cite{LT80, CNS81}.
This idea can be used for other problems admitting linear kernels in graph families admitting balanced separators:

\begin{enumerate}[(1)]
\item Compute a linear kernel $H$ for the original graph $G$, that is, $|V(H)|$ is at most $c_1 |OPT(H)|$ for some constant $c_1$. 
\item Decompose the kernel $H$ by recursively applying the separator algorithm and remove the separators until each resulting graph has at most $r$ vertices for some constant $r$. The union of all the separators has at most $c_2|V(H)| / \sqrt{r} = \epsilon |OPT(H)|$ vertices for $r$ chosen appropriately.
\item Solve the problem optimally for all the resulting graphs. 
\item Let $U_H$ be the union of all separators and all solutions of the resulting graphs. Lift $U_H$ to a solution $U_G$ for the original graph. 
\end{enumerate}

We obtain the following theorem, whose proof is provided in the appendix.
\begin{restatable}{theorem}{thmptas}
\label{thm:ptas}
There is an $O(n \log n)$ time PTAS for FVS in planar graphs.
\end{restatable}
More specifically, the running time is $O(2^cn+ n \log n)$ where $c=O(1/\epsilon^2)$ and $\epsilon$ is the error parameter.

\subsection{Heuristics}
In this subsection, we introduce some heuristics that can help improve the quality of the FVS solution. We first provide two heuristics that improve the quality of our PTAS solutions. Then we introduce a hybrid algorithm that combines the greedy method of the 2-approximation algorithm and the reduction rules of the kernelization algorithm. Finally, we use local search to improve the solution from any algorithm.
Our proposed Heuristic Approximation Scheme is a combination of the hybrid algorithm and the local search heuristic.

\subsubsection{Heuristic Improvements to PTAS}\label{sec:heu-ptas}
The solution from our PTAS may not be a minimal one, so 
we use the post-processing step from the 2-approximation algorithm to convert the final solution of the PTAS to a minimal one. This involves iterating through the vertices in the solution and trying to remove redundant vertices from the solution while maintaining feasibility. In fact, we only need to iterate through the vertices in separators, since vertices in the optimal solutions of small graphs are needed for feasibility. 

We additionally apply the kernelization algorithm right after we compute a separator in Step (2).  
Note that there is a decomposition tree corresponding to the decomposition step, where each node corresponds to a subgraph in the decomposition step. To apply the second heuristic, we need to record the whole decomposition tree with all the corresponding separators such that we can lift the solutions in the right order. 
For example, if we want to lift a solution for a subgraph $G_w$ corresponding to some node $w$ in the decomposition tree, we first need to lift all solutions for the subgraphs corresponding to the children of node $w$ in the decomposition tree.

\subsubsection{Hybrid Algorithm}\label{sec:hybrid}
We notice the cost for detecting applicable reduction rules for a candidate vertex is relatively low ($O(\log n)$ time), and each of these rules can reduce either the size of the graph or the size of the optimal solution. 
It is beneficial, therefore, to apply them as much as possible.  
When there are no applicable reductions in the current graph, we can remove a vertex greedily just as the 2-approximation algorithm does, and this will change the current graph such that there may be applicable reductions again.
Based on this idea, we propose a hybrid algorithm which interleaves the 2-approximation algorithm and the kernelization algorithm:
\begin{enumerate}[(i)]
\item Compute a temporary solution by repeating the following two steps until the graph is empty. 
\begin{enumerate}[(a)]
\item Apply reduction rules from the kernelization algorithm in order until there are no applicable rules.
\item Remove a vertex of highest degree from the graph and add it into the temporary solution.
\end{enumerate}
\item Lift the temporary solution to a feasible solution for the original graph and then compute a minimal solution by removing redundant vertices from this solution.
\end{enumerate}

The running time of the first step is similar to the running time of the kernelization algorithm since the greedy step can be seen as another ``rule'' added into the kernelization algorithm and it can be done in $O(\log n)$ time if we store the degree information in a binary search tree.
So the first step runs in $O(n \log n)$ time. 
The lifting step needs $O(n \log n)$ time, and to compute a minimal solution we need $O(n \log n)$ time as done in Becker and Geiger's 2-approximation algorithm.
So the total running time of our hybrid algorithm is $O(n \log n)$.

\subsubsection{Local Search}\label{sec:ls}
In our companion paper~\cite{LZ18}, we show that local search gives a PTAS for FVS in $H$-minor-free graphs. The algorithm is not practical, with running time $n^{O(1/\epsilon^2)}$.
We relax the conditions of the algorithm here.
Assume we are given a feasible FVS solution $U$ for a planar graph $G$, and we would like to improve this solution.
To achieve this goal, we propose a local search heuristic. Assume we have a fixed parameter tractable (FPT) algorithm {\it FA} for the FVS problem (either for planar graphs or for general graphs).
Then our local search with size $k$ consists of the following two steps.
\begin{enumerate}[(S1)]
\item Select a subset $X$ of size $k$ from $U$ randomly.
\item Run the algorithm {\it FA} on graph $G \setminus (U \setminus X)$. Stop {\it FA} if a solution is not found in a reasonable amount of time. If {\it FA} finds a solution $Y$ in graph $G \setminus (U \setminus X)$, then return $(U\setminus X) \cup Y$; otherwise return $U$. 
\end{enumerate}

The solution $(U \setminus X) \cup Y$ is a feasible solution for graph $G$ since
$(G \setminus (U \setminus X) ) \setminus Y$ is a forest and  $(G \setminus (U \setminus X) ) \setminus Y = G \setminus ((U \setminus X) \cup Y)$.

\section{Experiments}
In this section, we evaluate the performance of the algorithms described in the last section. We implemented those algorithms in \texttt{C++} and the code is compiled with \texttt{g++} (version 4.8.5) on the \texttt{CentOS} (version 7.3.1611) operating system. Our PTAS implementation is built on Boyer's implementation\footnote{http://jgaa.info/accepted/2004/BoyerMyrvold2004.8.3/planarity.zip} of Boyer and Myrvold's planar embedding algorithm~\cite{BM04}.
In our experiments, we also use the implementation\footnote{https://github.com/wata-orz/fvs} of Iwata and Imanishi for FVS in general graphs, which is implemented in \texttt{java} and includes a linear-time kernel~\cite{Iwata16} and a branch-and-bound based FPT algorithm~\cite{IWY16} for FVS in general graphs.
The \texttt{java} version in our machine is 1.8.0 and our experiments were performed on a machine with Intel(R) Xeon(R) CPU (2.30GHz) running \texttt{CentOS} (version 7.3.1611) operating system.

To test the algorithms, we collect five different classes of graphs:
\begin{itemize}
\item {\bf pace} are the planar graphs used in PACE (The Parameterized Algorithms and Computational Experiments Challenge) 2016 Track B: Feedback Vertex Set;
\item {\bf random} are random planar graphs generated by \texttt{LEDA} (version 6.4)~\cite{MNU97};
\item {\bf triangu} are triangulated random graphs generated by \texttt{LEDA}, whose outer faces are not triangulated;
\item {\bf grid} are rectangular grid graphs;
\item The graphs {\bf NY}, {\bf BAY}, {\bf COL}, {\bf NW}, {\bf CAL}, {\bf FLA}, {\bf LKS}, {\bf NE}, {\bf E} and {\bf W} are road networks used in the 9th DIMACS Implementation Challenge—Shortest Paths~\cite{DGJ08}.  We interpret each graph as a straight-line embedding and we add vertices whenever two edges intersect geometrically.
\end{itemize}
Since we are interested in the performance of the algorithms in large planar graphs, the synthetic graphs we generated have at least 450000 vertices. And the real network graphs have at least 260000 vertices.
Although the {\bf pace} graphs are smaller than that, we only use them to evaluate the 2-approximation algorithm since we can obtain optimal solutions of those small graphs.
All the detailed experimental results, including solution sizes and running time, are also provided in the appendices.

\subsection{The 2-Approximation Algorithm and Optimal Solution}\label{sec:2app-opt}

To evaluate Becker and Geiger's 2-approximation algorithm~\cite{BG96}, we compare its solution with the optimal solution on graphs up to size 2000000. The optimal solution is obtained by applying the kernelization algorithm first and then the FPT algorithm implemented by Iwata and Imanishi. 
For each test graph, we run Iwata and Imanishi's implementation for 30 seconds and stop it if it cannot terminate. 
Among all the test graphs, this method can solve 21 graphs of the 31 {\bf pace} graphs, 9 graphs of the 10 {\bf random} graphs.
Although we cannot solve the large rectangular grid graphs by this method, we can apply Luccio's~\cite{Luccio98} lower bound for the optimal solution in rectangular grids here. 
We observe that the solutions obtained by 2-approximation algorithm are very close to the optimal solutions for these test graphs. For {\bf pace} graphs, the 2-approximation algorithm can solve 14 graphs optimally and the difference between the two solutions is at most three. The approximation ratio over all these graphs is at most 1.143, and this ratio is at most 1.001 for the {\bf grid} graphs and at most 1.006 for these {\bf random} graphs.

\subsection{The PTAS and 2-Approximation Algorithm}\label{sec:ptas-2app}

Recall that the third step of our PTAS is to solve the problem on all the decomposed graphs optimally, which needs an exact algorithm for the problem.
The trivial exact algorithm that enumerates each possible vertex subset can only solve a graph of size about 25 in a few seconds, with which the PTAS may not be able to give competitive solutions for large graphs.
So we apply the exact algorithm described in the last subsection, which combines the kernelization algorithm and the FPT algorithm implemented by Iwata and Imanishi.
In the early stage of this work, this exact algorithm did not find a solution for a {\bf pace} graph of size 66 in 30 seconds.
So we first set as 60 the largest size $r$ of the decomposed graphs and let the FPT algorithm run for at most 15 seconds.
In this setting, all the decomposed graphs can be solved optimally in our experiments and we can evaluate the heuristics proposed for our PTAS.  
To do that, we compare three variants of our PTAS:
\begin{itemize}
\item the {\bf vanilla} variant is the vanilla version of our PTAS, for which no heuristic is applied;
\item the {\bf minimal} variant applies the post processing heuristic to our PTAS, which will remove redundant vertices in separators;
\item the {\bf optimized} variant applies both heuristics to our PTAS, which will apply kernelization algorithm whenever each separator is computed and removed during the decomposition step and always returns a minimal final solution.
\end{itemize}

The result is illustrated in Figure~\ref{fig:ptas}, where the solution size is normalized by the 2-approximation algorithm solutions. 
We observe that for the road network graphs, {\bf random} graphs and {\bf triangu} graphs, the {\bf optimized} variant provides the best solutions among the three algorithms, which implies the two heuristics both help improve the solutions.
However, for the {\bf grid} graphs, the {\bf minimal} variant gives the best results, which means the kernelization algorithm is not very helpful. We think this is because a large rectangular grid graph itself is already a $4k$-kernel by the lower bound for the optimal solution~\cite{Luccio98}, and the kernelization algorithm can only remove four vertices from such graphs.
For the {\bf random} graphs (not pictured in Figure~\ref{fig:ptas}), the improvement from the two heuristics is mild, but the obtained solutions are already very close to the optimal solutions, that is the differences are less than 60 vertices when the solutions have more than 190000 vertices. 
We also find that the post processing heuristic will not affect the running time by much, while the kernelization heuristic can increase the running time at most by a factor of 5.

\begin{figure}

\centering
\includegraphics[scale = 0.52]{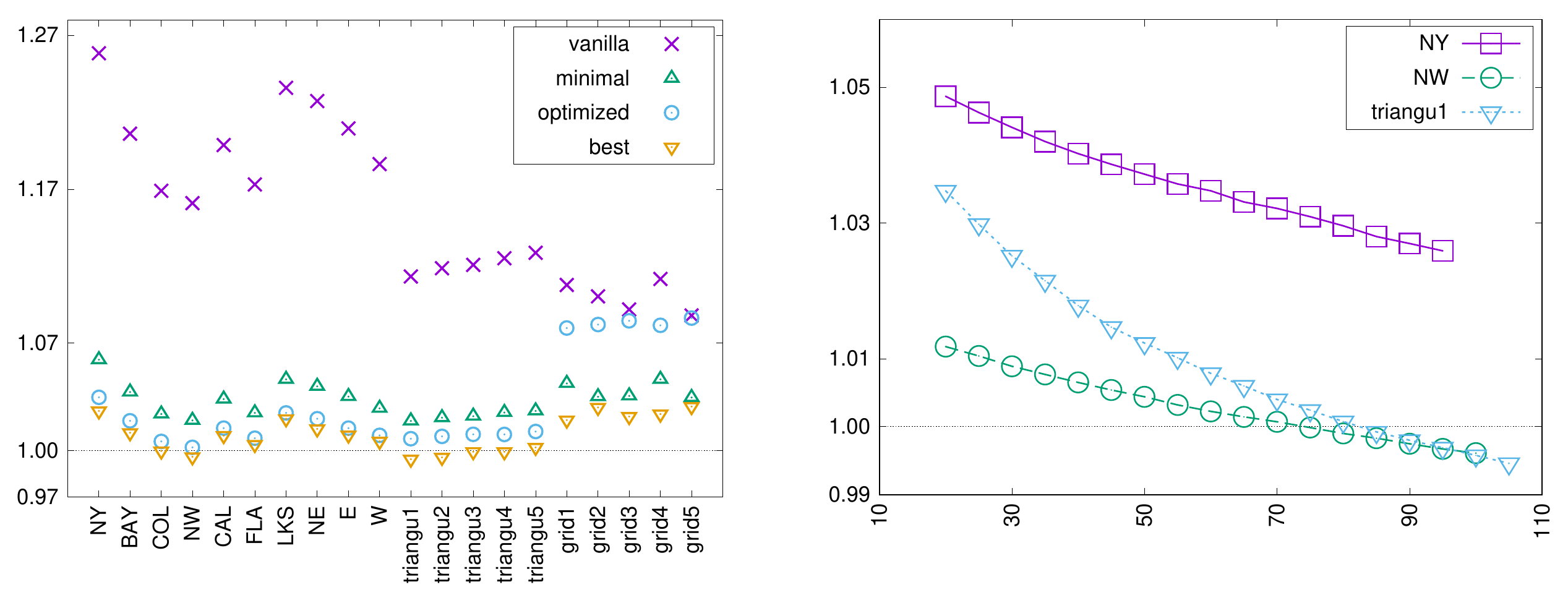}
\caption{Results of our PTAS implementation. The Y axis represents the solution size normalized by the solutions of 2-approximation algorithm.
Left: solutions of the three variants of our PTAS with $r=60$ and the best solutions of our PTAS with largest possible values of $r$. 
Right: effect of parameter $r$ on the solution sizes where X axis represents the value of $r$.
}
\label{fig:ptas}
\end{figure}

Recall that the largest size $r$ of the decomposed graphs is the only parameter in our PTAS.
Now we analyze the effect of this parameter on the performance of our PTAS.
Since the {\bf optimized} variant works best for most of the test graphs, we focus on its performance affected by parameter $r$. 
We start with $r=20$ and each time increase it by 5 until our implementation cannot compute a feasible solution, which is caused by the fact that the FPT implementation cannot solve some decomposed graphs of size $r$ in the given time.
The result is shown in Figure~\ref{fig:ptas}.
We can see in the figure that our implementation can solve the instance for relatively large value of $r$, and the solution is improved when the parameter $r$ increases.
This is because when $r$ is bigger, the total size of the separators, which is the error part of our PTAS, is smaller.
Since we set a time limit for the FPT implementation, the effect of parameter $r$ on the running time is not significant, although we observed a mildly increasing tendency on the running time when $r$ increases.

Now we can compare our PTAS with the 2-approximation algorithm. 
Based on the above results, we know the largest value of $r$ may be different for different graphs. So to get the best result, we start with $r=60$ and each time increase it by 5 until we cannot find feasible solution.
Since different variants work best for different graph classes, we choose the best variant for each graph class, that is, we apply the {\bf minimal} variant for {\bf grid} graphs, and the {\bf optimized} variant for other graphs.
The largest value $r$ we find varies from 80 to 125 for our test graphs.
The final result is plotted in Figure~\ref{fig:ptas} marked as {\bf best}.
We can see that although our PTAS implementation is competitive with the 2-approximation algorithm on all graphs, it cannot outperform the 2-approximation solution on most road network graphs.
The reason is that the subgraphs for which we can solve the problem optimally are still not large enough,
which results in a fact that the separator fraction in the solution is large.  
On the {\bf grid} graphs, the 2-approximation algorithm outperforms our PTAS by about 2 percent.
On four of the five {\bf triangu} graphs, our PTAS is slightly better than the 2-approximation algorithm, where the difference for each graph is less than 1 percent.
While on the {\bf random} graphs the improvement is mild, the final solutions are very close to the optimal.
We also find the running time of our PTAS is much longer than the 2-approximation algorithm.
Specifically, the running time of the 2-approximation algorithm varies from a few seconds to a few minutes for our test graphs, while our PTAS needs a few hours to finish for the large graphs, where most time is spent on running the FPT algorithm.

\subsection{Heuristic Approximation Scheme}\label{sec:heu-2app}
We evaluate the performance of the following two algorithms and compare them with the 2-approximation algorithm ({\bf 2approx}).
\begin{itemize}
\item The {\bf hybrid} algorithm combines the kernelization reduction rules and the greedy step of the {\bf 2approx} to find a solution, and then lifts it to a minimal feasible solution for the original graph.
\item {\bf HAS} first computes a solution by the {\bf hybrid} algorithm and then applies the local search heuristic with an FPT algorithm to improve the solution quality.
\end{itemize}

In our following experiments, the FPT algorithm used in the local search consists of two parts: Bonamy and Kowalik's linear kernelization algorithm and Iwata and Imanishi's implementation, which is run for at most 15 seconds.

Before showing the final comparison of these algorithms, we first need to optimize their parameters.  
For the {\bf hybrid} algorithm, there is a potential parameter: how often should the kernelization reduction rules be applied?
We evaluated different values from 1 to 100 for this frequency parameter to understand its effect on the solution size.  
We find that for the road network graphs the difference is at most 0.2 percent of the solution size and for the other test graphs the difference is less than 0.01 percent.
So we only consider the road network graphs to optimize this parameter.
Although we did not find an optimal value for this parameter that could always give smallest solution, we can avoid some values that always give larger solutions.
For this goal, we choose the frequency as 41 in our later experiments, that means we apply the reduction rules after removing 41 vertices greedily.

\begin{figure}
\centering
\includegraphics[scale = 0.52]{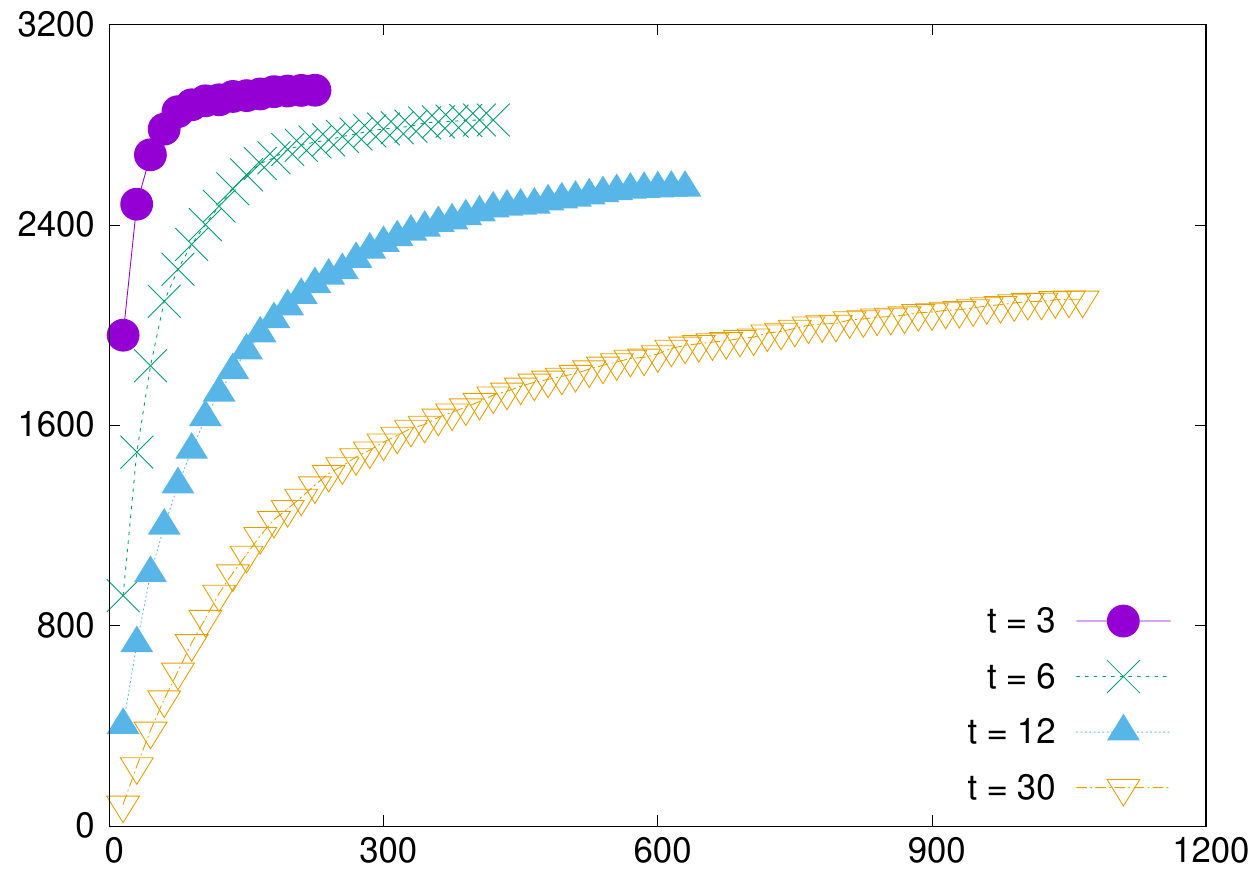}
\includegraphics[scale = 0.52]{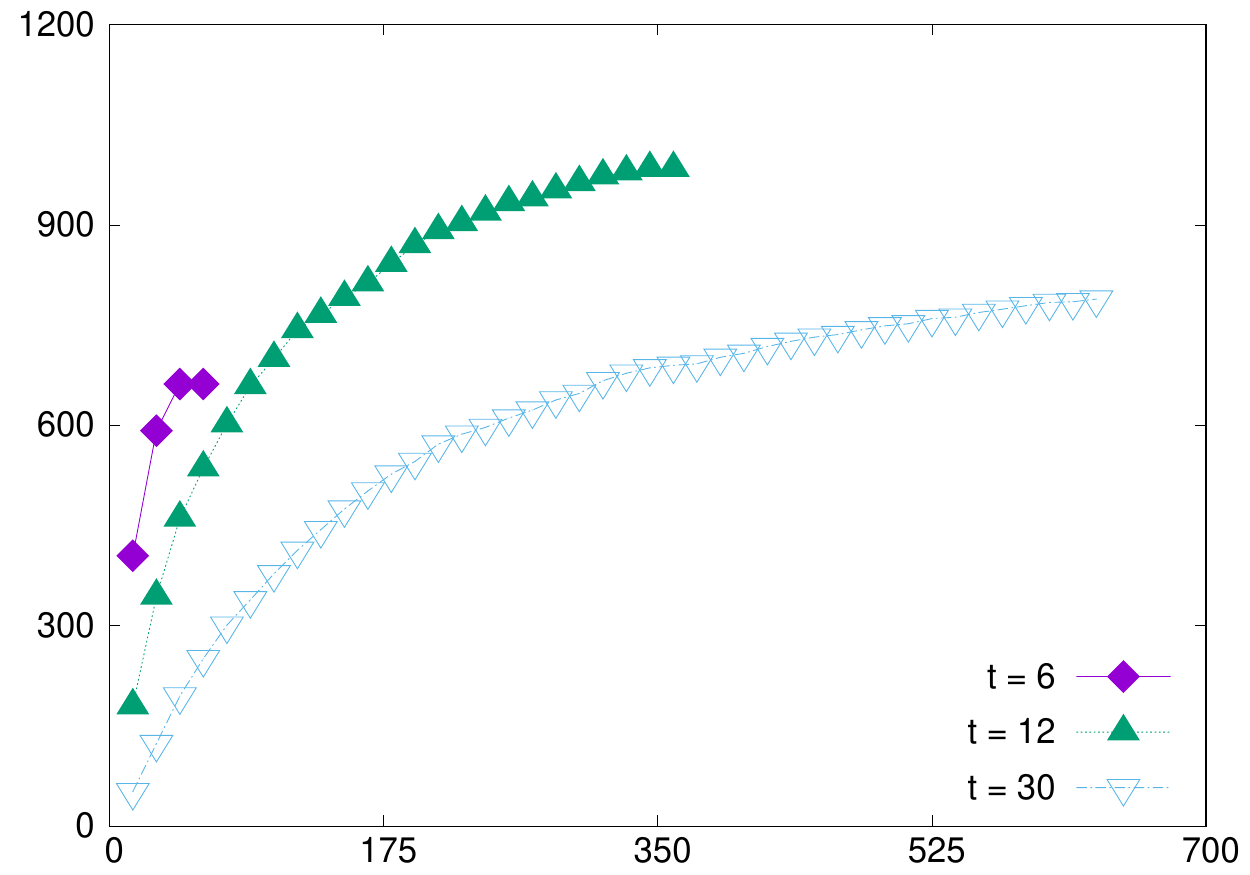}
\caption{Local search improvement affected by different values of $t$. 
The X axis represents the number of local search iterations and the Y axis represents the number of vertices the initial solution is improved by.
Left: results for graph {\bf triangu1}. Right: results for graph {\bf COL}. 
}
\label{fig:ls1}
\end{figure}

For the local search heuristic in {\bf HAS}, the only parameter we need to optimize is the size $k$ of the random set.
We choose this parameter as a $1/t$ fraction of the solution size for an integer $t$,  
so that we only need to optimize the value of $t$.
We evaluated four different values (3, 6, 12, 30) for $t$ and terminate the process when there is no improvement for twelve consecutive search rounds.
Figure~\ref{fig:ls1} shows the results for two graphs.
We can see in the figures that the total improvement is not always monotone with $t$ in this range.
For the synthetic graphs like {\bf triangu1}, the total improvement for $t=3$ is the largest, while for most road network graphs, the largest total improvement is obtained for $t=12$.
The former phenomenon can be explained by our companion work~\cite{LZ18}, which shows local search is a PTAS for FVS in minor-free graphs so the solution will be better when the size of the random set replaced is larger.  
However, this will not hold when the FPT algorithm cannot solve the problem in a reasonable amount of time for large random set, which corresponds to the smaller $t$, and then the improvement will be limited.
This could explain the results for the road network graphs.
Moreover, while local search with larger value of $k$ tends to give better improvements, the number of local search iterations is bigger for smaller value of $k$, which implies the improvement is consecutive and stable.
Based on these observations, we will iterate through the values of $t$ in an increasing order to maximize the improvement in our later experiments, that is, we will increase the value of $t$ by 3 when there are six consecutive search rounds that find no improvement. The range of $t$ is still from 3 to 30.

To better understand the local search heuristic, we illustrate the improvement fraction and running time fraction distributed on different values of $t$ in Figure~\ref{fig:ls2}, where the running time is represented by the number of local search iterations.  We observe that for the synthetic graphs, represented by {\bf triangu5} in the figure, the total improvement comes from the search with $t=3$, while for the road network graphs, the improvement distributes on different small values of $t$.  The distribution of the running time has the similar tendency as the improvement.

We notice that the local search heuristic in {\bf HAS} can give us a ``PTAS behavior'', that is we can obtain better solutions if we spend more time doing so.
So one natural question to ask is how long can that improvement process last? 
To answer this, we tried looping the values of $t$ in the range [3, 30], and found that only one iteration over the range is enough, and that additional iterations only give minor improvement.

\begin{figure}
\centering
\includegraphics[scale = 0.52]{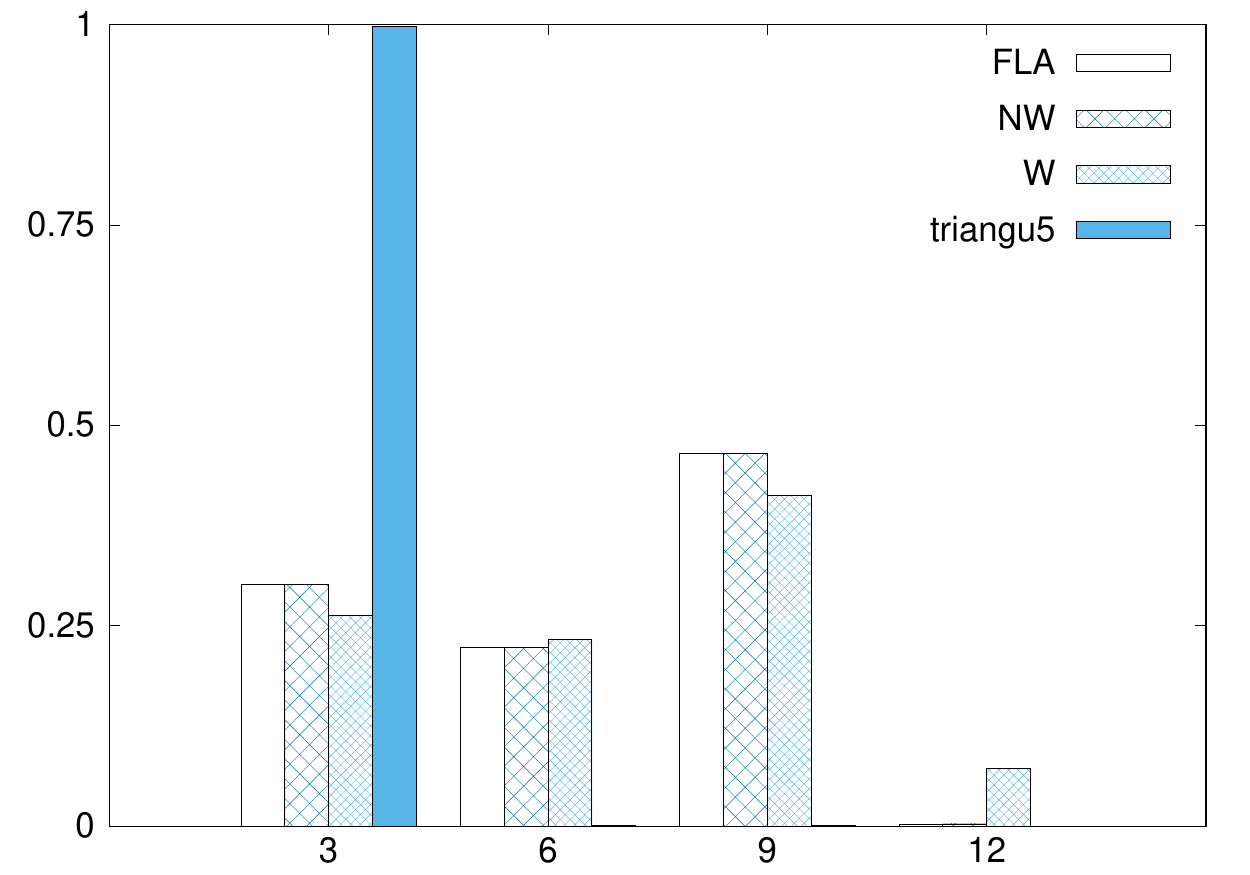}
\includegraphics[scale = 0.52]{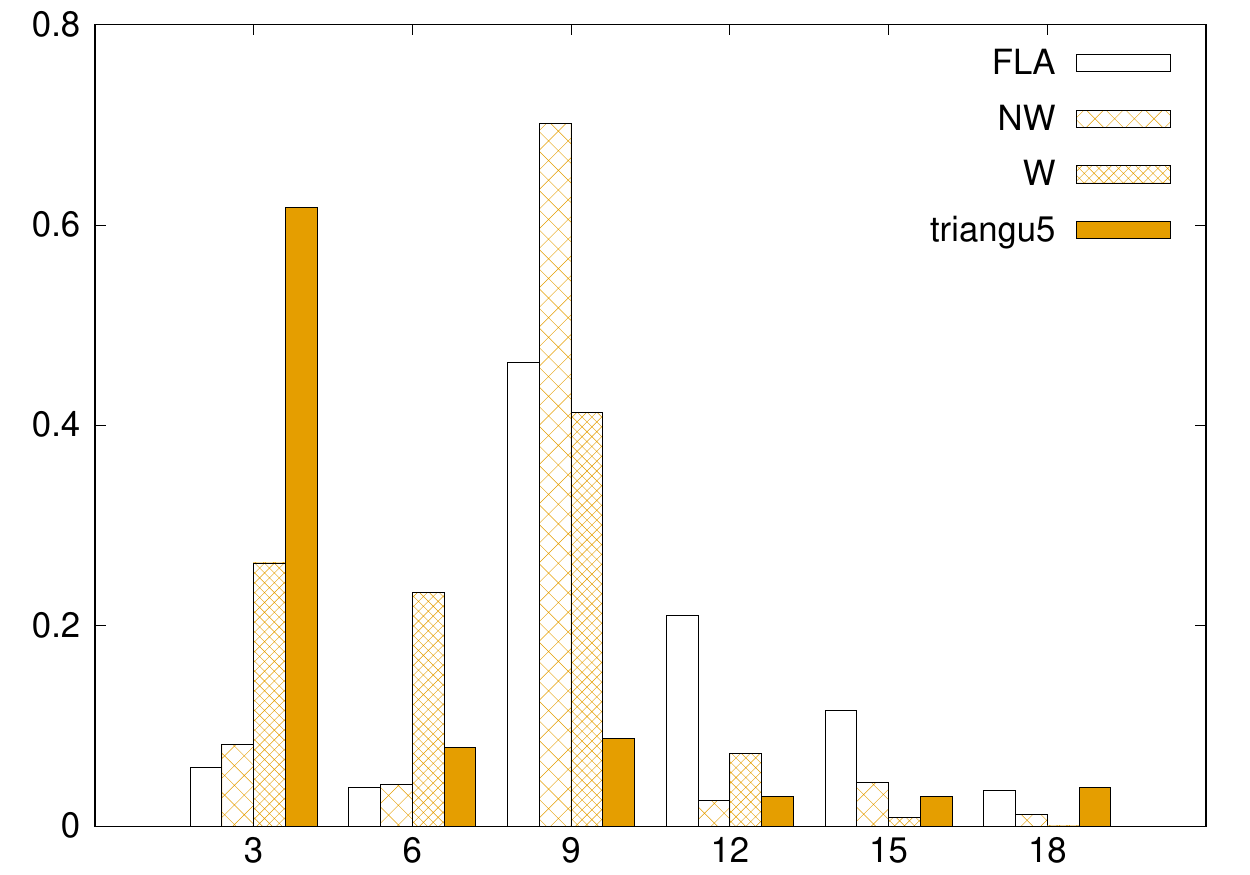}
\caption{Local search improvement and number of iterations distributed on different values of $t$.  The X axis is the value of $t$.  Left: the Y axis represents the fraction of total improvement. Right: the Y axis represents the fraction of the total iterations. 
} 
\label{fig:ls2}
\end{figure}

Now we can set the parameters and compare these algorithms. 
For each graph, we run {\bf HAS} five times, each of which is run with a different random seed, and compute the average of the five solutions.
Figure~\ref{fig:total} illustrates the results where each solution size is normalized by the {\bf 2approx} solution.
We can see in the figure that for most graphs, {\bf hybrid} finds better solutions than {\bf 2approx}, and {\bf HAS} gives the best solutions. 
The improvement of {\bf HAS} for all road network graphs is more than 5 percent, which could be over 30000 vertices for large graphs like {\bf W}.
For all {\bf triangu} graphs the improvement of {\bf HAS} is at least 3 percent. 
For {\bf random} graphs the improvement is not very significant, this is because the final solutions are already very close to the optimal solutions.
Since the {\bf 2approx} works very well on the {\bf grid} graphs as shown before, it is hard to outperform it on these graphs though our heuristics is able to find competitive solutions on these graphs.
Although the improvement of {\bf HAS} is significant, its running time is relatively long compared with the other two algorithms.
For example, both of the {\bf 2approx} and {\bf hybrid} can terminate in a few minutes for graph {\bf W}, but {\bf HAS} needs more than 35 hours to terminate for this graph.
For this, people may need to balance the running time and the solution quality by setting a proper number of local search iterations.

\begin{figure}
\centering
\includegraphics[scale = 0.8]{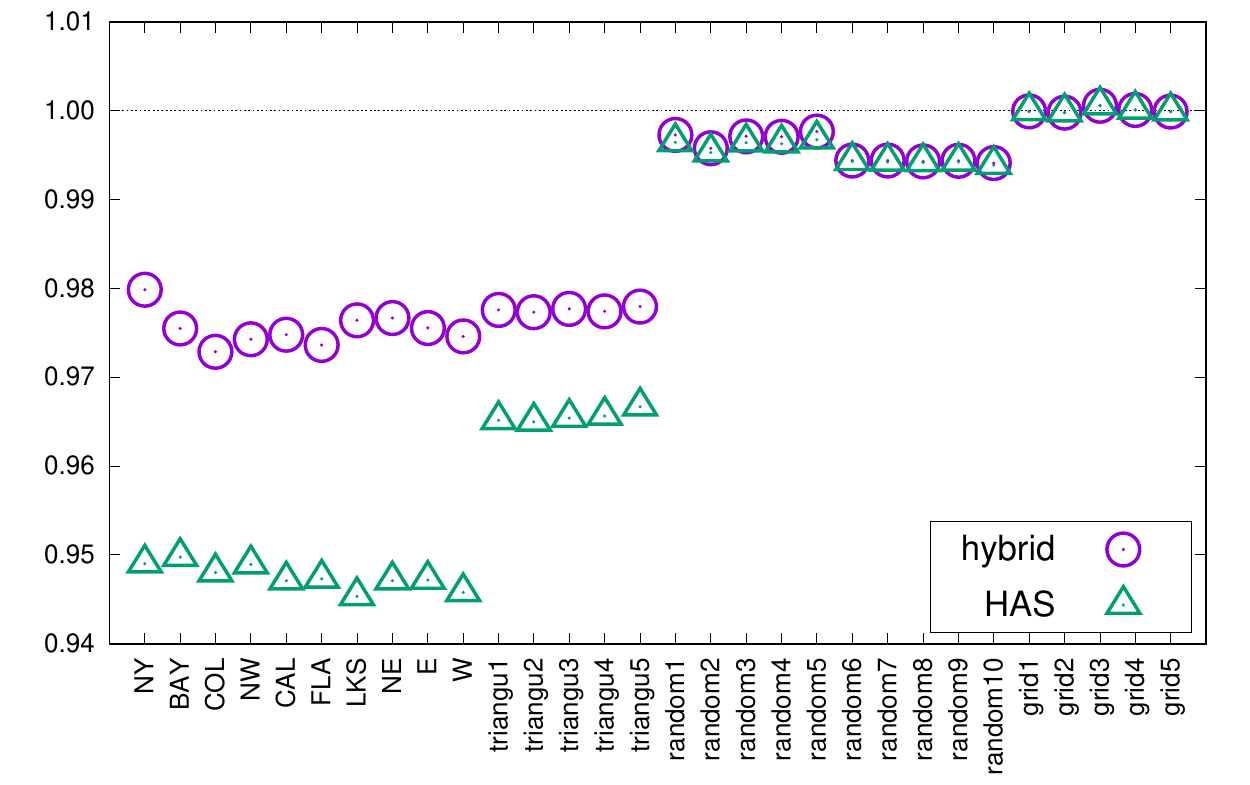}
\caption{Results of {\bf hybrid} and {\bf HAS}. The Y axis represents the solution sizes normalized by the solutions of the 2-approximation algorithms. The solution of {\bf HAS} is the average over the five runs with different random seeds.}
\label{fig:total}
\end{figure}

\section{Conclusions}
We proposed an $O(n \log n)$ time PTAS and a heuristic algorithm for the minimum feedback vertex set problem in planar graphs. We also evaluated their performance and compare them with the 2-approximation algorithm. Our results show that our PTAS is competitive with the 2-approximation algorithm, and our heuristic algorithm can outperform the 2-approximation on both synthetic graphs and real-world graphs. 
We remark that we can also obtain an $O(n \log n)$ time PTAS for FVS in planar graphs by applying Baker's technique~\cite{Baker94} for the linear kernel, but then we need to compute a branch decomposition and handle the dynamic programming in the branch decomposition as done in previous PTAS engineering works~\cite{TM11, BFKM17}. It will be interesting to see if this kind of PTAS can outperform the 2-approximation and our heuristic algorithm.

Although we focus on FVS in this work, we believe the ideas behind our heuristics can also work for other problems. For example, the idea of combining reduction rules and another approximation algorithm can also be applied to other problems like dominating set and vertex cover. Similarly, our local search heuristic can also be generalized to new problems if there are FPT algorithms for them and it will be interesting to see how do they work on different problems.

\bibliography{fvs_ref}

\newpage
\appendix

\section{Omitted Proofs}

\begin{proof}[Proof of Obervation~\ref{obs:y}]
\begin{figure}
\centering
\includegraphics[scale=1.2]{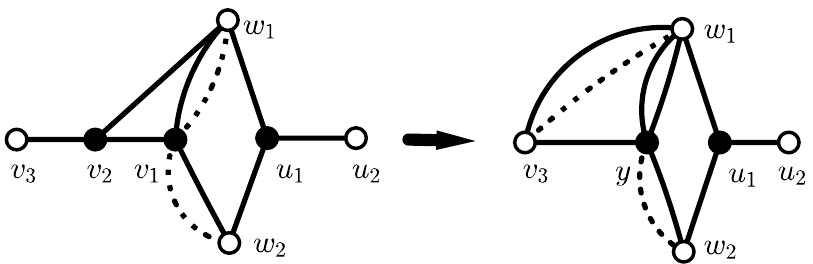}
\caption{Reduction rule 8 replaces the left subgraph with the right subgraph. All dashed edges are optional, and all incident edges of black vertices are drawn as solid or dashed, while the white vertices may have other edges in the graph not drawn in the figure.}
\label{fig:r8}
\end{figure}

Rule 8 is illustrated in Figure~\ref{fig:r8}, which will not modify the size of the optimal solution.
We will show the new vertex $y$ can be replaced by the vertex $v_1$ in the resulting graph after applying this rule.
Let $G$ and $G'$ be the graph before and after applying Rule 8, and let $S'$ be an FVS solution for $G'$.  
We only need to prove the following:
if $y \in S'$, then $(S'\setminus \{ y \}) \cup \{ v_1 \}$ is an FVS solution for $G$; otherwise $S'$ is an FVS solution for $G$.

Assume $y \in S'$. Then we claim that any FVS solution for $G' \setminus \{ y \}$ is also an FVS solution for $G \setminus \{v_1 \}$.
This is because we can obtain $G' \setminus \{ y \}$ from $G \setminus \{v_1\}$ by applying Rule 3 (that is, if a vertex has degree two and its incident edges are not parallel edges, we remove this vertex and add an edge between its two neighbors) to vertex $v_2$, which will not affect the size of the optimal solution.
Then we know $G \setminus (\{ v_1 \} \cup (S' \setminus \{y \}))$ is a forest, which implies $(S' \setminus \{y \}) \cup \{ v_1\}$ is an FVS for $G$. 

Now assume $y \notin S'$. Then we know $w_1 \in S'$ since there are parallel edges between $y$ and $w_1$. We claim that any FVS solution for $G' \setminus \{w_1 \}$ that does not contain $y$ is also an FVS for $G\setminus \{w_1\}$. This is because after applying Rule 3 to $v_2$ in $G\setminus \{w_1\}$ the resulting graph is the same as $G'\setminus \{w_1\}$ except that the label for the vertex $v_1$ is $y$ in $G' \setminus \{w_1\}$.
Since $y \notin S'$, we know $S' \setminus \{w_1\}$ is an FVS for $G \setminus \{w_1\}$, which implies $S'$ is an FVS for $G$.
\end{proof}

We need the following lemma to prove Theorem~\ref{thm:ptas}.

\begin{lemma}\label{lem:x}
Given any FVS solution $U_H$ for a linear kernel $H$ obtained by Bonamy and Kowalik's algorithm, we can obtain an FVS solution $U_G$ for the original graph $G$ such that $|U_G| - |U_H| \le |OPT(G)| - |OPT(H)|$.
\end{lemma}
\begin{proof}
We can classify the reduction rules in Bonamy and Kowalik's kernelization algorithm into three types according to their effects on the optimal solution.
\begin{itemize}
\item Do not affect the optimal solution, such as removing vertices of degree one.
\item Remove some vertices from the graph that must be in the optimal solution $OPT(G)$, such as removing vertices with self-loops.
\item Add some new vertices into the graph without changing the size of the optimal solution, such as replacing a subgraph with another subgraph that has some new vertices.
\end{itemize}
When lifting the solution of $H$ to that of $G$, we have their corresponding effects:
\begin{itemize}
\item Do not change the current solution.
\item Add some new vertices to the current solution.
\item Replace some vertices in the current solution with other vertices without increasing its size.
\end{itemize} 
The rules of the third type will maintain the size of optimal solution unchanged, so we know $x = |OPT(G)| - |OPT(H)|$ is equal to the total number of vertices added by the second type of rules.
Since the rules of the third type cannot increase the size of the solution during the lifting step, we know the size difference after applying a reverse rule of the third type is non-positive.
That is, when we apply any reverse rule of the third type to a solution $U_0$ and obtain a new solution $U_1$, we have $|U_1| - |U_0| \le 0$.
Let $y$ be the sum of size differences over all third type rules applied during the lifting step.
Then we know $y$ is also non-positive.
Note that the rules of the second type will only add vertices that is not in the kernel, so it contributes the same vertices to $OPT(G)$ as to any other solution $U_G$.  
Therefore, we know $|U_G| - |U_H| = x + y$, which implies the lemma.
\end{proof}

Now we are ready to prove Theorem~\ref{thm:ptas}.
\thmptas*
\begin{proof} 
We first give a bound for the size of the final solution.
For an integer $i > 0$, let $H_i$ be a resulting graph after the decomposition step. Note that these graphs are vertex-disjoint.
Let $OPT(H)$ and $OPT(G)$ be an optimal solution for $H$ and $G$ respectively.
We know $OPT(H) \cap H_i$ is a solution for FVS in $H_i$ since $H_i$ is a subgraph of $H$ and $OPT(H)$ is an optimal solution for $H$.
So we have $|OPT(H_i)| \le |OPT(H) \cap H_i|$ and then 
\begin{align}
\sum_{i>0} |OPT(H_i)| \le \sum_{i>0} |OPT(H) \cap H_i| \le |OPT(H)|.
\label{equ:a}
\end{align}
Let $S$ be the union of all separators found in the second step and $n$ be the size of $V(H)$.
Recall that parameter $r$ is the largest size of $H_i$.
Lipton and Tarjan~\cite{LT80} showed that the size of $S$ is at most $c_2 n/\sqrt{r}$ for some constant $c_2$.
If we choose $r = c_1^2 c_2^2 / \epsilon^2$, then we have 
\begin{align}
|S| \le c_2 n / \sqrt{r} = \epsilon n/c_1 \le \epsilon |OPT(H)|.
\label{equ:b}
\end{align}
Since $U_H$ is the union of $S$ and $OPT(H_i)$ for all $i>0$, by combining (\ref{equ:a}) and (\ref{equ:b}), we have 
\[|U_H| = |S| + \sum_{i>0} |OPT(H_i)| \le (1+\epsilon) |OPT(H)|.\]
Since the kernelization algorithm can only decrease the size of the optimal solution, we have $|OPT(H)| \le |OPT(G)|$.
By Lemma~\ref{lem:x}, we have $|U_G| - |U_H| \le |OPT(G)| - |OPT(H)|$.
Then we obtain the size bound for $U_G$: 
\[|U_G| \le |U_H| + |OPT(G)| - |OPT(H)| \le  (1+\epsilon) |OPT(G)|. \]

Bonamy and Kowalik~\cite{BK16} showed that a linear kernel for planar FVS can be constructed in $O(n \log n)$ deterministic time.
Each balanced separator can be computed in linear time by Lipton and Tarjan's algorithm~\cite{LT79}, so we can finish the second step in $O(n \log n)$ time as done in~\cite{LT80}.
The third step can be finished in $O(2^{c}n)$ time since each subgraph has size at most $c=r=O(1/\epsilon^2)$.
And the last step can be done in $O(n \log n)$ time as we described in Section~\ref{sec:kernel}.
So the total running time of this algorithm is $O(2^{c}n + n \log n)$ where $c = O(1/\epsilon^2)$.  
\end{proof}

\section{Detailed Experimental Results}

\begin{table}
\caption{Compare the 2-approximation algorithm solutions and the optimal solutions. The italic numbers are lower bounds from~\cite{Luccio98}.}
\begin{tabular}{p{2cm} p{1.2cm} p{1.2cm} p{1.2cm} p{1.2cm} p{2.2cm}}
\toprule 
graph & vertices & edges & 2approx & opt & approx ratio\\
\midrule
{\bf pace1} & 49 & 107 & 15 & 15 & 1.0 \\
{\bf pace2} & 118 & 179 & 18 & 18 & 1.0 \\
{\bf pace3} & 62 & 78 & 7 & 7 & 1.0 \\
{\bf pace4} & 118 & 179 & 18 & 18 & 1.0 \\
{\bf pace5} & 59 & 104 & 16 & 16 & 1.0 \\
{\bf pace6} & 70 & 85 & 8 & 8 & 1.0 \\
{\bf pace7} & 48 & 64 & 6 & 6 & 1.0 \\
{\bf pace8} & 74 & 92 & 8 & 8 & 1.0 \\
{\bf pace9} & 67 & 83 & 9 & 8 & 1.125 \\
{\bf pace10} & 90 & 103 & 8 & 7 & 1.143 \\
{\bf pace11} & 55 & 81 & 12 & 11 & 1.091 \\
{\bf pace12} & 110 & 147 & 16 & 15 & 1.067 \\
{\bf pace13} & 66 & 127 & 24 & 21 & 1.143 \\
{\bf pace14} & 153 & 177 & 12 & 12 & 1.0 \\
{\bf pace15} & 149 & 193 & 16 & 16 & 1.0 \\
{\bf pace16} & 73 & 95 & 10 & 10 & 1.0 \\
{\bf pace17} & 45 & 64 & 8 & 8 & 1.0 \\
{\bf pace18} & 145 & 186 & 17 & 16 & 1.063 \\
{\bf pace19} & 158 & 189 & 15 & 15 & 1.0 \\
{\bf pace20} & 61 & 78 & 8 & 7 & 1.143 \\
{\bf pace21} & 4960 & 9462 & 898 & 898 & 1.0 \\
{\bf grid1} & 450000 & 1796400 & 149527 & {\it 149401} & 1.001 \\ 
{\bf grid2} & 600000 & 2396800 & 199552 & {\it 199468} & 1.001 \\ 
{\bf grid3} & 1000000 & 3996000 & 332669 & {\it 332668} & 1.0 \\ 
{\bf grid4} & 1680000 & 6714800 & 559252 & {\it 559134} & 1.001 \\ 
{\bf grid5} & 2100000 & 8394200 & 699285 & {\it 699034} & 1.001 \\ 
{\bf random1} & 699970 & 2000000 & 193601 & 192902 & 1.004 \\
{\bf random2} & 1197582 & 3000000 & 285550 & 284195 & 1.005 \\
{\bf random3} & 1399947 & 4000000 & 387216 & 385813 & 1.004 \\
{\bf random4} & 1999760 & 5600000 & 540550 & 538524 & 1.004 \\
{\bf random6} & 873280 & 1200000 & 87300 & 86796 & 1.006 \\
{\bf random7} & 1061980 & 1500000 & 111973 & 111324 & 1.006 \\
{\bf random8} & 1227072 & 1700000 & 125374 & 124635 & 1.006 \\
{\bf random9} & 1520478 & 2200000 & 167709 & 166737 & 1.006 \\
{\bf random10} & 2050946 & 3300000 & 270981 & 269315 & 1.006 \\
\bottomrule
\end{tabular}
\label{tab:opt}
\end{table}

\begin{table}
\caption{Summary of the solution sizes of different PTAS variants and the 2-approximation algorithm. The results for the three variants (``vanilla'', ``minimal'' and ``optimized'') of the PTAS are computed for $r=60$. The results for ``best'' are computed for different values of $r$ by different variants of PTAS depending on the graphs.  For {\bf grid} graphs, ``minimal'' variant is applied for ``best'' results, and for other graphs ``optimized'' variant is applied. }
\begin{tabular}{p{1.5cm} p{1.1cm} p{1.1cm} p{1.1cm} p{1.1cm} p{1.15cm} p{1.1cm} p{1.1cm} }
\toprule 
graph          & vertices & edges & vanilla & minimal & optimized & best & 2approx \\
\midrule
{\bf NY}       & 264953 & 366250 & 52487 & 44178 & 43159 & 42790 & {\bf 41709} \\
{\bf BAY}      & 322694 & 400233 & 41264 & 35518 & 34877 & 34612 & {\bf 34211} \\
{\bf COL}      & 437294 & 524437 & 45831 & 40142 & 39445 & {\bf 39201} & 39205 \\
{\bf NW}       & 1214463 & 1423402 & 111149 & 97627 & 95951 & {\bf 95362} & 95735 \\
{\bf FLA}      & 1074167 & 1351411 & 147632 & 128941 & 126869 & 126361 & {\bf 125841} \\
{\bf CAL}      & 1898842 & 2331204 & 230712 & 198924 & 195303 & 194379 & {\bf 192465} \\
{\bf LKS}      & 2763392 & 3407840 & 342592 & 290065 & 284034 & 282972 & {\bf 277196} \\
{\bf NE}       & 1528387 & 1941840 & 222569 & 188936 & 185115 & 183916 & {\bf 181336} \\
{\bf E}        & 3608115 & 4372928 & 414132 & 354444 & 347418 & 345891 & {\bf 342371} \\
{\bf W}        & 6286759 & 7608797 & 705851 & 611351 & 600898 & 598667 & {\bf 594943} \\
{\bf triangu1} & 600000 & 1799963 & 260475 & 238502 & 235796 & {\bf 232695} & 233958 \\
{\bf triangu2} & 800000 & 2399961 & 348305 & 318058 & 314257 & {\bf 310034} & 311349 \\
{\bf triangu3} & 1000000 & 2999955 & 434793 & 396644 & 392045 & {\bf 387690} & 387908 \\
{\bf triangu4} & 1200000 & 3599966 & 521141 & 474683 & 468072 & {\bf 462865} & 463157 \\
{\bf triangu5} & 1400000 & 4199957 & 606425 & 551194 & 544006 & 538496 & {\bf 537270} \\
{\bf random1} & 699970 & 2000000 & 192970 & 192925 & 192925 & {\bf 192907} & 193601 \\
{\bf random2} & 1197582 & 3000000 & {\bf 284195} & {\bf 284195} & {\bf 284195} & {\bf 284195} & 285550 \\
{\bf random3} & 1399947 & 4000000 & 385928 & 385837 & 385866 & {\bf 385818} & 387216 \\
{\bf random4} & 1999760 & 5600000 & 538620 & 538556 & 538573 & {\bf 538529} & 540550 \\
{\bf random5} & 2199977 & 6400000 & 617940 & 617761 & 617768 & {\bf 617681} & 619710 \\
{\bf random6} & 873280 & 1200000 & {\bf 86796} & {\bf 86796} & {\bf 86796} & {\bf 86796} & 87300 \\
{\bf random7} & 1061980 & 1500000 & {\bf 111324} & {\bf 111324} & {\bf 111324} & {\bf 111324} & 111973 \\
{\bf random8} & 1227072 & 1700000 & {\bf 124635} & {\bf 124635} & {\bf 124635} & {\bf 124635} & 125374 \\
{\bf random9} & 1520478 & 2200000 & {\bf 166737} & {\bf 166737} & {\bf 166737} & {\bf 166737} & 167709 \\
{\bf random10} & 2050946 & 3300000 & {\bf 269315} & {\bf 269315} & {\bf 269315} & {\bf 269315} & 270981 \\
{\bf grid1} & 450000 & 1796400 & 165650 & 156071 & 161462 & 152515 & {\bf 149527} \\
{\bf grid2} & 600000 & 2396800 & 219589 & 206563 & 215920 & 205189 & {\bf 199552} \\
{\bf grid3} & 1000000 & 3996000 & 363242 & 344549 & 360783 & 340078 & {\bf 332669} \\
{\bf grid4} & 1680000 & 6714800 & 621736 & 585211 & 604861 & 572608 & {\bf 559252} \\
{\bf grid5} & 2100000 & 8394200 & 760797 & 723507 & 759568 & 719410 & {\bf 699285} \\
\bottomrule
\end{tabular}
\label{tab:ptas-size}
\end{table}

\begin{table}
\caption{Summary of the running time of different variants of our PTAS. The parameter $r$ is set as 60 for all PTAS variants.}
\begin{tabular}{p{2.0cm} p{1.1cm} p{1.1cm} p{1.7cm} p{1.7cm} p{1.7cm} p{1.1cm} }
\toprule 
graph          & vertices & edges & vanilla & minimal & optimized & 2approx \\
\midrule
{\bf NY}       & 264953 & 366250 & 5m 45s & 5m 39s & 13m 46s & 9s \\
{\bf BAY}      & 322694 & 400233 & 4m 25s & 4m 17s & 10m & 9s  \\
{\bf COL}      & 437294 & 524437 & 5m 8s & 4m 56s & 10m 46s & 13s  \\
{\bf NW}       & 1214463 & 1423402 & 13m 21s & 12m 49s & 26m 2s & 37s \\
{\bf FLA}      & 1074167 & 1351411 & 16m 38s & 16m 37s & 34m 18s & 33s \\
{\bf CAL}      & 1898842 & 2331204 & 35m 36s & 33m 2s & 1h 12m 20s & 58s \\
{\bf LKS}      & 2763392 & 3407840 & 1h 26s & 1h 1m 46s & 2h 48m 44s & 1m 31s \\
{\bf NE}       & 1528387 & 1941840 & 41m 13s & 43m 32s & 1h 53m 1s & 48s \\
{\bf E}        & 3608115 & 4372928 & 1h 39m 17s & 1h 40m 7s & 3h 34m 25s & 1m 52s \\
{\bf W}        & 6286759 & 7608797 & 3h 44m 38s & 3h 43m 26s & 7h 39m 21s & 3m 15s \\
{\bf triangu1} & 600000 & 1799963 & 1h 6m 23s & 55m 39s & 1h 56m 32s & 55s \\
{\bf triangu2} & 800000 & 2399961 & 2h 28m 7s & 1h 57m 15s & 4h 15m 24s & 1m 7s \\
{\bf triangu3} & 1000000 & 2999955 & 4h 17m 57s & 3h 41m 50s & 6h 47m 30s & 1m 22s \\
{\bf triangu4} & 1200000 & 3599966 & 6h 42m 30s & 5h 48m 12s & 9h 59m 24s & 1m 39s \\
{\bf triangu5} & 1400000 & 4199957 & 9h 20m 14s & 8h 14m 59s & 14h 16m 45s & 1m 58s \\
{\bf random1} & 699970 & 2000000 & 50m 40s & 53m 18s & 1h 22m 10s & 1m 20s \\
{\bf random2} & 1197582 & 3000000 & 1h 19m 18s & 1h 22m 27s & 1h 37m 57s & 1m 58s \\
{\bf random3} & 1399947 & 4000000 & 2h 42m 55s & 2h 25m 9s & 2h 5m 3s & 2m 34s \\
{\bf random4} & 1999760 & 5600000 & 4h 1m 48s & 4h 30m 10s & 3h 32m 19s & 3m 53s \\
{\bf random5} & 2199977 & 6400000 & 4h 57m 53s & 5h 38m 16s & 4h 39m 10s & 4m 24s \\
{\bf random6} & 873280 & 1200000 & 14m 9s & 16m 8s & 13m 8s & 58s \\
{\bf random7} & 1061980 & 1500000 & 19m 54s & 22m 33s & 18m 10s & 1m 5s \\
{\bf random8} & 1227072 & 1700000 & 21m 44s & 24m 29s & 19m 41s & 1m 14s \\
{\bf random9} & 1520478 & 2200000 & 30m 52s & 35m 13s & 28m 9s & 1m 33s \\
{\bf random10} & 2050946 & 3300000 & 1h 1m & 1h 10m 29s & 55m 41s & 2m 20s \\
{\bf grid1} & 450000 & 1796400 & 25m 41s & 24m 28s & 1h 13m 18s & 24s \\
{\bf grid2} & 600000 & 2396800 & 35m 24s & 34m 48s & 1h 58m 24s & 29s \\
{\bf grid3} & 1000000 & 3996000 & 1h 11m 9s & 1h 6m 49s & 4h 56m 29s & 48s \\
{\bf grid4} & 1680000 & 6714800 & 2h 56m 47s & 2h 18m 7s & 10h 7m 46s & 1m 25s \\
{\bf grid5} & 2100000 & 8394200 & 3h 48m 16s & 3h 0m 23s & 15h 1m 48s & 1m 44s \\
\bottomrule
\end{tabular}
\label{tab:ptas-time}
\end{table}

\begin{table}
\caption{Summary of the solution sizes of different heuristic algorithms. HAS (avg) is the solution size averaged over five runs of HAS with different random seeds. HAS (min) is the minimum size among the five values.  The improv value is computed as  1-HAS(avg)/2approx.}
\begin{tabular}{p{1.5cm} p{1.1cm} p{1.1cm} p{1.1cm} p{1.1cm} p{1.1cm} p{1.3cm} p{1.1cm} }
\toprule 
graph          & vertices & edges & 2approx & hybrid & HAS (avg) & HAS (min) & improv \\
\midrule
{\bf NY}       & 264953 & 366250 & 41709 & 40868 & 39582 & {\bf 39537} & 0.051 \\
{\bf BAY}      & 322694 & 400233 & 34211 & 33372 & 32492 & {\bf 32446} & 0.050 \\
{\bf COL}      & 437294 & 524437 & 39205 & 38141 & 37167 & {\bf 37141} & 0.052 \\ 
{\bf NW}       & 1214463 & 1423402 & 95735 & 93272 & 90844 & {\bf 90816} & 0.051 \\
{\bf FLA}      & 1074167 & 1351411 & 125841 & 122521 & 119211 & {\bf 119163} & 0.053 \\
{\bf CAL}      & 1898842 & 2331204 & 192465 & 187610 & 182282 & {\bf 182191} & 0.053 \\
{\bf LKS}      & 2763392 & 3407840 & 277196 & 270651 & 262041 & {\bf 261917} & 0.055 \\ 
{\bf NE}       & 1528387 & 1941840 & 181336 & 177101 & 171742 & {\bf 171619} & 0.053 \\ 
{\bf E}        & 3608115 & 4372928 & 342371 & 333997 & 324286 & {\bf 324056} & 0.053 \\ 
{\bf W}        & 6286759 & 7608797 & 594943 & 579822 & 562687 & {\bf 562606} & 0.054 \\ 
{\bf triangu1} & 600000 & 1799963 & 233958 & 228711 & 225802 & {\bf 225791} & 0.035 \\ 
{\bf triangu2} & 800000 & 2399961 & 311349 & 304288 & 300443 & {\bf 300430} & 0.035 \\ 
{\bf triangu3} & 1000000 & 2999955 & 387908 & 379251 & 374485 & {\bf 374468} & 0.035 \\ 
{\bf triangu4} & 1200000 & 3599966 & 463157 & 452692 & 447236 & {\bf 447206} & 0.034 \\ 
{\bf triangu5} & 1400000 & 4199957 & 537270 & 525416 & 519361 & {\bf 519229} & 0.033 \\ 
{\bf random1} & 699970 & 2000000 & 193601 & 193071 & 192908 & {\bf 192904} & 0.004 \\ 
{\bf random2} & 1197582 & 3000000 & 285550 & 284340 & 284198 & {\bf 284195} & 0.005 \\ 
{\bf random3} & 1399947 & 4000000 & 387216 & 386095 & 385818 & {\bf 385815} & 0.004 \\ 
{\bf random4} & 1999760 & 5600000 & 540550 & 538958 & 538533 & {\bf 538526} & 0.004 \\ 
{\bf random5} & 2199977 & 6400000 & 619710 & 618238 & 617674 & {\bf 617671} & 0.003 \\ 
{\bf random6} & 873280 & 1200000 & 87300 & 86809 & 86802 & {\bf 86796} & 0.006 \\ 
{\bf random7} & 1061980 & 1500000 & 111973 & 111343 & 111326 & {\bf 111324} & 0.006 \\ 
{\bf random8} & 1227072 & 1700000 & 125374 & 124655 & 124641 & {\bf 124636} & 0.006 \\ 
{\bf random9} & 1520478 & 2200000 & 167709 & 166762 & 166739 & {\bf 166739} & 0.006 \\ 
{\bf random10} & 2050946 & 3300000 & 270981 & 269378 & 269319 & {\bf 269318} & 0.006 \\ 
{\bf grid1} & 450000 & 1796400 & 149527 & {\bf 149511} & {\bf 149511} & {\bf 149511} & 0.001 \\ 
{\bf grid2} & 600000 & 2396800 & 199552 & {\bf 199506} & {\bf 199506} & {\bf 199506} & 0.001 \\ 
{\bf grid3} & 1000000 & 3996000 & {\bf 332669} & 332847 & 332847 & 332847 & -0.001 \\ 
{\bf grid4} & 1680000 & 6714800 & {\bf 559252} & 559298 & 559298 & 559298 & -0.001 \\ 
{\bf grid5} & 2100000 & 8394200 & 699285 & {\bf 699195} & {\bf 699195} & {\bf 699195} & 0.001 \\ 
\bottomrule
\end{tabular}
\label{tab:others-size}
\end{table}

\begin{table}
\caption{Summary of the running time of different heuristic algorithms. HAS (avg) is the running time averaged over five runs of HAS with different random seeds. HAS (min) is the running time for the HAS run with minimum solution size.}
\begin{tabular}{p{2.0cm} p{1.1cm} p{1.1cm} p{1.1cm} p{1.2cm} p{1.7cm} p{1.7cm} }
\toprule 
graph          & vertices & edges & 2approx & hybrid & HAS (avg) & HAS (min)\\
\midrule
{\bf NY}       & 264953 & 366250 & 9s & 1m 28s & 1h 6m 30s & 1h 12m 33s \\
{\bf BAY}      & 322694 & 400233 & 9s  & 1m 9s & 54m 53s & 1h 8m 2s \\
{\bf COL}      & 437294 & 524437 & 13s  & 1m 19s & 51m 22s & 1h 10m 49s \\
{\bf NW}       & 1214463 & 1423402 & 37s & 3m 28s & 2h 18m 40s & 2h 24m 32s \\
{\bf FLA}      & 1074167 & 1351411 & 33s & 4m 12s & 3h 57m 25s & 4h 7m 32s \\
{\bf CAL}      & 1898842 & 2331204 & 58s & 7m 32s & 6h 10m 28s & 6h 12m 30s \\
{\bf LKS}      & 2763392 & 3407840 & 1m 31s & 11m 38s & 13h 21m 22s & 14h 35m 33s \\
{\bf NE}       & 1528387 & 1941840 & 48s & 7m 10s & 5h 52m 16s & 6h 49m 51s \\
{\bf E}        & 3608115 & 4372928 & 1m 52s & 14m 23s & 16h 49m 22s & 16h 27m 17s \\
{\bf W}        & 6286759 & 7608797 & 3m 15s & 23m 46s & 35h 41m 39s & 38h 10m 5s \\
{\bf triangu1} & 600000 & 1799963 & 55s & 1m 46s & 1h 32m 37s & 1h 39m 7s \\
{\bf triangu2} & 800000 & 2399961 & 1m 7s & 1m 58s & 2h 11m 39s & 2h 35m \\
{\bf triangu3} & 1000000 & 2999955 & 1m 22s & 2m 20s & 2h 33m 5s & 2h 42m 50s \\
{\bf triangu4} & 1200000 & 3599966 & 1m 39s & 2m 55s & 3h 15m 9s & 3h 54m 12s \\
{\bf triangu5} & 1400000 & 4199957 & 1m 58s & 3m 30s & 4h 38m & 3h 23m 16s \\
{\bf random1} & 699970 & 2000000 & 1m 20s & 1m 54s & 6m 54s & 7m 31s \\
{\bf random2} & 1197582 & 3000000 & 1m 58s & 2m 26s & 7m 47s & 8m 57s \\
{\bf random3} & 1399947 & 4000000 & 2m 34s & 3m 31s & 15m 5s & 15m 24s \\
{\bf random4} & 1999760 & 5600000 & 3m 53s & 5m 12s & 21m 2s & 21m 53s \\
{\bf random5} & 2199977 & 6400000 & 4m 24s & 5m 58s & 30m 15s & 32m 43s \\
{\bf random6} & 873280 & 1200000 & 58s & 1m 17s & 1m 45s & 2m 12s \\
{\bf random7} & 1061980 & 1500000 & 1m 5s & 1m 31s & 2m 4s & 1m 55s \\
{\bf random8} & 1227072 & 1700000 & 1m 14s & 1m 42s & 2m 25s & 2m 13s \\
{\bf random9} & 1520478 & 2200000 & 1m 33s & 2m 17s & 3m 17s & 3m 18s \\
{\bf random10} & 2050946 & 3300000 & 2m 20s & 3m 29s & 5m 10s & 5m 56s \\
{\bf grid1} & 450000 & 1796400 & 24s & 3m 49s & 1h 1m 16s & 1h 2m 26s \\
{\bf grid2} & 600000 & 2396800 & 29s & 5m 19s & 1h 18m 18s & 1h 18m 54s \\
{\bf grid3} & 1000000 & 3996000 & 48s & 8m 48s & 1h 55m 10s & 1h 56m 43s \\
{\bf grid4} & 1680000 & 6714800 & 1m 25s & 14m 47s & 3h 1m 44s & 3h 4m 5s \\
{\bf grid5} & 2100000 & 8394200 & 1m 44s & 17m 44s & 3h 52m 47s & 3h 42m 23s \\
\bottomrule
\end{tabular}
\label{tab:others-time}
\end{table}

\end{document}